\documentclass[11 pt]{amsart}
\usepackage{graphicx}

\usepackage{amsthm,amsmath,latexsym,amssymb,amsfonts,amscd,hyperref}

\usepackage[all]{xy}
\newtheorem{theorem}{Theorem}[section]
\newtheorem{lemma}[theorem]{Lemma}

\theoremstyle{definition}

\theoremstyle{remark}
\newtheorem{remark}[theorem]{Remark}

\numberwithin{equation}{section}

\begin{document}
\vspace{2cm}
 \title[A simple construction of associative deformations ]{A simple construction of associative deformations}  

\author{Alexey A. Sharapov}
\address{Physics Faculty, Tomsk State University, Lenin ave. 36, Tomsk 634050, Russia}
\email{sharapov@phys.tsu.ru}
\thanks{The first author was supported in part by RFBR
Grant No. 16-02-00284 A and by Grant No. 8.1.07.2018 from ``The Tomsk State University competitiveness improvement programme''.}

\author{Evgeny D. Skvortsov}
\address{Albert Einstein Institute, 
Am M\"{u}hlenberg 1, D-14476, Potsdam-Golm, Germany}
\address{Lebedev Institute of Physics, 
Leninsky ave. 53, 119991 Moscow, Russia}
%\address{Lebedev Institute of Physics, 
%Leninsky ave. 53, 119991 Moscow, Russia}
\email{evgeny.skvortsov@aei.mpg.de}
\thanks{The second author was partially supported by RSF Grant No. 18-72-10123 in association with Lebedev Physical Institute.}

%    General info

%\date{January 1, 2001 and, in revised form, June 22, 2001.}


\begin{abstract}
We propose a simple approach to formal deformations of associative algebras. It exploits the machinery of multiplicative coresolutions of an associative algebra $A$ in the category of $A$-bimodules. Specifically, we show that certain first-order deformations of $A$ extend to all orders and we derive explicit recurrent formulas determining this extension.  
In physical terms, this  may be regarded as the deformation quantization of noncommutative Poisson structures on $A$. 
\end{abstract}

\maketitle
%%%%%%%%%%%%%%%%%%%%%%%%%%%%%%%%%%%%%%%%%%%%%%%%%%%%%%%%%%%%%%%%%%%%%%%
\section{Introduction}
%%%%%%%%%%%%%%%%%%%%%%%%%%%%%%%%%%%%%%%%%%%%%%%%%%%%%%%%%%%%%%%%%%%%%%%

The algebraic deformation theory is at the heart of several branches of modern mathematical physics: quantum groups, deformation quantization, noncommutative geometry and field theory are just a few examples.  For associative algebras the deformation theory is known to be closely related to the Hochschild cohomology \cite{GSh}. In particular, the infinitesimal deformations of an associative algebra $A$ are classified by the elements of the second cohomology group $HH^2(A,A)$, while the third group $HH^3(A,A)$ controls the obstructions to integrability of infinitesimal deformations. As was first observed by Gerstenhaber  \cite{Gerst}, the groups $HH^\bullet(A,A)$ carry a rich algebraic structure, namely, that of a graded Poisson algebra.  Among other things this allows one to transfer the notion of a Poisson structure from smooth manifolds to  noncommutative algebras \cite{BG}, \cite{Xu}. 

By definition, a Poisson structure on an associative algebra $A$ is an element $\Pi\in HH^2(A,A)$ whose Gerstenhaber bracket with itself vanishes, that is, $[\Pi,\Pi]=0$. The problem of noncommutative deformation quantization can now be formulated as follows \cite{Tang}: given a Poisson structure $\Pi$ on $A$, construct an associative $\ast$-product in $A[[t]]$ such that 
\begin{equation}\label{star}
a\ast b =ab+\sum_{n=1}^\infty t^n\mu_n (a,b) \qquad \forall a,b\in A
\end{equation}
and the Hochschild cohomology class of $[\mu_1]$ is equal to $\Pi$.  As particular cases this includes the usual (i.e.,`commutative') deformation quantization of Poisson manifolds \cite{BFFLS} as well as the noncommutative deformation theory of Refs. \cite{Pinczon-1}, \cite{Nadaud-1}. 

While the problem of quantizing smooth Poisson manifolds has been completely solved by Kontsevich \cite{Konts}, only a few explicit examples of genuine noncommutative Poisson structures and their deformation quantization are available in the literature.  Most of them are related to the smash-product algebras $A=B\rtimes \Gamma $, with $\Gamma$ being a finite group of automorphisms of  an algebra $B$.  The case  that $B$ is the algebra of smooth functions on a $\Gamma$-manifold $M$ was thoroughly studied in \cite{NPPT}, \cite{HT}, \cite{HOT}. The noncommutative algebra $C^\infty(M)\rtimes \Gamma$ is known to be a good substitute for the commutative algebra $C^\infty(M)^\Gamma$ of $\Gamma$-invariant functions whenever the action of $\Gamma$ on $M$ is not free, so that the quotient space $M/\Gamma$ is singular \cite{Co}.  In \cite{HT}, it was shown that, in addition to the usual Poisson bivectors on $M/\Gamma$, the algebra $C^\infty(M)\rtimes \Gamma$
admits noncommutative Poisson structures with supports on codimention $2$, fixed-point submanifolds $M^\gamma\subset M$, $\gamma\in \Gamma$. If  one takes  $M$ to be a symplectic vector space $V$ endowed with a faithful action of a symplectic reflection group $\Gamma\subset \mathrm{Sp}(V)$, then the corresponding quantum algebra belongs to the class of symplectic reflection algebras introduced in \cite{EG}. 
One more example, where $B$ is the polynomial Weyl algebra  on two generators and $\Gamma= \mathbb{Z}_2$,  was considered in \cite{Pinczon-1}. 

It should be noted that in most known examples of noncommutative deformation quantization the corresponding  $\ast$-products were not obtained in an explicit form. 
Their  existence followed indirectly from the Poincar\'e--Birkhoff--Witt property of the corresponding quantum algebras.  Perhaps the only exception is the paper \cite{Nadaud-2}, where a Moyal-type formula was derived  for the deformation quantization of $C^\infty(\mathbb{R}^2)\rtimes \mathbb{Z}_n$.  

In this paper, we propose a new approach to noncommutative deformation quantization. The main advantage of our method is that, similar to the Fedosov deformation quantization \cite{Fedosov}, it provides explicit recurrent formulas for $\ast$-products, and not just existence theorems. Let us briefly explain the basic idea behind our approach.   

As we have already mentioned, the deformations of an associative algebra $A$ are governed by the Hochschild cohomology groups whose formal definition is the following:
$$
HH^\bullet (A,A) =\mathrm{Ext}_{A^e}^\bullet (A,A)\,.
$$
Here $A^e=A\otimes A^{\mathrm{op}}$ is the enveloping algebra of $A$. To compute the $\mathrm{Ext}$-groups on the right one may use either projective or injective resolutions of the algebra 
$A$ viewed as a module over $A^e$. Let us denote these resolutions by $\mathcal{A}_P$ and $\mathcal{A}_J$:

$$
\xymatrix { \cdots\ar[r]&\mathcal{A}_P^2\ar[r]  &\mathcal{A}_P^1\ar[r]  &\mathcal{A}_P^0\ar[r]  &A\ar[r]& 0\,,}
$$
$$
\xymatrix {  0\ar[r]& A\ar[r] &\mathcal{A}_J^0 \ar[r]& \mathcal{A}_J^1\ar[r]& \mathcal{A}_J^2\ar[r] &\cdots \,.}
$$
It is the standard result of homological algebra (see e.g. \cite{ML}) that 
\begin{equation}\label{Ext}
H^\bullet (\mathrm{Hom}_{A^e}(\mathcal{A}_P, A))\simeq \mathrm{Ext}^\bullet_{A^e } (A,A)\simeq H^\bullet (\mathrm{Hom}_{A^e}(A, \mathcal{A}_J))\,.
\end{equation}
There is also a definition symmetric with respect to both the arguments of $\mathrm{Ext}$-functor:
$$
\mathrm{Ext}^\bullet_{A^e } ({A},A)\simeq H^\bullet (\mathrm{Hom}_{A^e}(\mathcal{A}_P, \mathcal{A}_J))\,.
$$

The most popular way of computing the Hochschild cohomology is through the projectives. Among various projective resolutions of the algebra $A$ there exists  the  standard resolution $\mathcal{B}(A)$, called also  the bar-resolution. This appears naturally in many specific problems including deformation theory. Despite its theoretical importance, the bar-resolution does not  help much  in practical computations with infinite-dimensional algebras.  In order to compute the Hochschild cohomology of a given algebra $A$ one has to find a special projective resolution $\mathcal{A}_P$.  In the context of deformation quantization all the computations with $\mathcal{A}_P$ should then be followed by the construction of a  quasi-isomorphism from the complex $\mathrm{Hom}_{A^e}(\mathcal{A}_P,A)$ to the standard bar-complex $\mathrm{Hom}_{A^e}(\mathcal{B}(A),A)$ to get an explicit formula for the $\ast$-product.  Generally this last step may not be an easy task.
Some explicit examples of quasi-isomorphisms between the Koszul and bar-complexes can be found in \cite{HT}, \cite{BGHHW}, \cite{Riv}, \cite{SW-2}. 

Rel. (\ref{Ext}) suggests  also an alternative way to compute the Hochschild cohomology, namely, through the injectives. It is easy to see that 
\begin{equation}
    \mathrm{Hom}_{A^e}(A,\mathcal{A}_J)\simeq \mathcal{A}_J^A\,,
\end{equation}
where $\mathcal{A}_J^A\subset \mathcal{A}_J$ is the submodule of $A$-invariants:
$$
\mathcal{A}_J^A=\{a\in \mathcal{A}_J\;|\; ab=ba\,,\,  \forall b\in A\}\,.
$$
Hence, the Hochschild cohomology of $A$ is isomorphic to the cohomology of the subcomplex $\mathcal{A}_J^A$, i.e.,
\begin{equation}\label{isom}
HH^\bullet(A,A)\simeq H^\bullet(\mathcal{A}_J^A)\,.
\end{equation}
Having in mind the deformation problem for the algebra $A$, it is quite natural to look for injective resolutions $\mathcal{A}_J$ that are differential graded algebras, and not just $A^e$-modules. In Sec. \ref{Def}, we show that certain elements of $H^2(\mathcal{A}^A_J)$  give rise to a formal deformation of $A$; in so doing, the corresponding $\ast$-product (\ref{star}) is recurrently  defined in terms of the Hochschild differential and the contracting homotopy operator for the  resolution $\mathcal{A}_J$. Moreover, the requirement of injectivity is not actually needed for our construction: in many interesting cases one can use non-injective coresolutions to produce nontrivial deformations. In Sec. \ref{Examples}, we illustrate  the general method by several examples of commutative and noncommutative algebras.

%%%%%%%%%%%%%%%%%%%%%%%%%%%%%%%%%%%%%%%%%%%%%%%%%%%%%%%%%%%%%
\section{Preliminaries}
\label{sec-1}

Throughout this paper, $k$ is an arbitrary field of characteristic zero and all unadorned tensor products $\otimes$ and Homs are taken over $k$. 

By a differential graded $k$-algebra (DG-algebra for short) we mean a non-negatively graded, unital, associative algebra  $\mathcal{A}=\bigoplus_{n\geq 0}\mathcal{A}^n$  over $k$ endowed with a differential $d:\mathcal{A}^n\rightarrow \mathcal{A}^{n+1}$ such that 
$$ \mathcal{A}^n\cdot \mathcal{A}^m\subset \mathcal{A}^{n+m}\,,\qquad 
d(ab)=(da)b+(-1)^{|a|}adb\,, \qquad d^2=0\,,
$$
where $|a|$ is the degree of a homogeneous element $a\in \mathcal{A}^{|a|}$. Notice that $\mathcal{A}^0$ is always a (non-differential) subalgebra in $\mathcal{A}$.

Given a DG-algebra $\mathcal{A}$, we can define the bicomplex  $C^{\bullet,\bullet}(\mathcal{A},\mathcal{A})$ of $k$-vector spaces 
$$
C^{n,m}(\mathcal{A},\mathcal{A})=\{f\in \mathrm{Hom}(\mathcal{A}^{\otimes n},\mathcal{A})\;|\; f \,\mbox{is a linear map of degree } m\}
$$
equipped with the pair of differentials
$$
d:C^{n,m}(\mathcal{A},\mathcal{A})\rightarrow C^{n,m+1}(\mathcal{A},\mathcal{A})\,,\qquad \delta:C^{n,m}(\mathcal{A},\mathcal{A})\rightarrow C^{n+1,m}(\mathcal{A},\mathcal{A})\,.
$$
For any cochain
$f: \mathcal{A}^{\otimes n}\rightarrow\mathcal{ A}$ we set
$$
(d f)(a_1\otimes \cdots\otimes a_n)= d f(a_1\otimes\cdots\otimes a_n)-\sum_{i=1}^n (-1)^{\epsilon_i} f(a_1\otimes\cdots\otimes da_i\otimes\cdots\otimes a_n)\,,
$$
$$
(\delta f)(a_1\otimes \cdots\otimes a_n)= -(-1)^{(|a_1|+1)|f|}a_1 f(a_2\otimes \cdots\otimes a_n)
$$
$$
-\sum_{i=2}^n (-1)^{\epsilon_i} f(a_1\otimes\cdots \otimes a_{i-1}a_{i}\otimes\cdots \otimes a_n)
$$
$$
+(-1)^{\epsilon_n}f(a_1\otimes\cdots\otimes a_{n-1})a_n\,.
$$
Here  $\epsilon_i=|f|+|a_1|+\cdots+|a_{i-1}|-i+1$ and $|f|=m+n$ is the total degree of the cochain $f$.  
 It follows from the definition that
$$
d^2=\delta^2=0\,,\qquad d\delta+\delta d=0
$$
and we can define the total differential $D=\delta+d$ that increases the total degree by $1$. The Hochschild cohomology groups $HH^\bullet(\mathcal{A},\mathcal{A})$ of the DG-algebra $\mathcal{A}$ with coefficients in itself  are now defined to be the cohomology of the total complex 
$$\mathcal{C}^{\bullet}(\mathcal{A},\mathcal{A}) =\mathrm{Tot}\,C^{\bullet,\bullet}(\mathcal{A},\mathcal{A})\,.$$

The space of Hochschild cochains $\mathcal{C}^{\bullet}(\mathcal{A},\mathcal{A})$ has the structure of graded  Lie algebra with respect to the Gerstenhaber bracket 
$$
[f,g]=f\circ g-(-1)^{(|f|+1)(|g|+1)}g\circ f\,,
$$
where 
$$
f\circ g =\sum_{i=0}^{n-1}(-1)^{(|g|+1)\sum_{j=1}^i(|a_j|+1)} f(a_1\otimes \cdots\otimes a_i\otimes g(a_{i+1}\otimes \cdots\otimes a_{i+1+m})\otimes \cdots \otimes a_{m+n-1})\,.
$$
This bracket satisfies the standard properties 
$$
[f,g]=-(-1)^{(|f|+1)(|g|+1)}[g,f]\,,
$$
$$
[[f,g],h]=[f,[g,h]]-(-1)^{(|f|+1)(|g|+1)}[g,[f,h]]\,,
$$
$$
d [f,g]=[d f,g]+(-1)^{(|f|+1)}[f,d g]\,.
$$
In addition,
$$
\delta f=[m,f]\,,\qquad m(a_1,a_2)=(-1)^{|a_1|}a_1a_2\,,
$$
and
$$
\delta [f,g]=[\delta f,g]+(-1)^{(|f|+1)}[f,\delta g]\,.
$$
As is seen, the coboundary operator  $D=\delta + d$ of the total complex differentiates  the Gerstenhaber bracket, so that we can speak of the differential graded Lie algebra $(\mathcal{C}^{\bullet}(\mathcal{A},\mathcal{A}), D)$. Clearly, the Gerstenhaber bracket descends to the cohomology making the $k$-vector space $HH^\bullet(\mathcal{A},\mathcal{A})$ into a graded Lie algebra.  

The usual definitions of the Hochschild cohomology and Gerstenhaber bracket for nongraded algebras are obtained from the above  formulas  just assuming  $\mathcal{A}$ to be concentrated in zero degree, i.e., $\mathcal{A}=\mathcal{A}^0$. Then the total degree is given by the first degree of ${C}^{\bullet,\bullet}$ and $D=\delta$ becomes the usual Hochschild differential. 
 
One of the motivations to introducing the Gerstenhaber bracket comes from deformation theory. Let us rewrite  (\ref{star}) in the form   $a\ast b=ab+\mu(a,b)$, where the $2$-cochain $\mu\in \mathcal{C}^2(A[[t]],A[[t]])$ describes the formal deformation of the $k[[t]]$-algebra $A[[t]]$. Then one can see that the $\ast$-product is associative iff the cochain $\mu$ satisfies the Maurer--Cartan equation 
$$
\delta \mu=-\frac12 [\mu,\mu]\,.
$$
Expanding this equation in powers of $t$, we get the sequence of equations 
$$
\delta \mu_1=0\,,\qquad \delta \mu_2=-\frac12[\mu_1,\mu_1]\,,\qquad \ldots 
$$
The first equation identifies $\mu_1$ as a $2$-cocycle representing an element of $HH^2(A,A)$. According to  the second equation this cocyle squares to zero up to coboundary, that is, defines a Poisson structure on $A$. At $n$-th order we face the equation  
\begin{equation}\label{mupsi}
\delta \mu_n=\psi_n\,, \qquad \psi_n=-\frac12\sum_{k=1}^{n-1}[\mu_k,\mu_{n-k}]\,.
\end{equation}
It follows from the definition that the $3$-cochain $\psi_n$ is $\delta$-closed provided all the previous $n-1$ equations are satisfied. 
Thus, the $(n-1)$-th order deformation can be extended to the next order whenever the cocycle $\psi_n$ is trivial. The extension, if exists, is not unique as we are free to add to $\mu_n$ any $2$-cocycle, e.g. a trivial one. If $HH^3(A,A)=0$, then the obstruction space is empty and we conclude immediately that any infinitesimal deformation can be integrated to a global one. In the general case, however,  the existence of a solution to Eq. (\ref{mupsi}) depends not only on the algebra $A$ itself, but also on the choice of particular  solutions for $\mu_1,\ldots,\mu_{n-1}$. This makes the whole problem of constructing  global deformations 
extremely difficult: no general method is known for  solving  Eqs. (\ref{mupsi}).

%%%%%%%%%%%%%%%%%%%%%%%%%%%%%%%%%%%%%%%%%%%%%%%%%%%%%%%%%%%%%%%%%%%%%%%%%%%
\section{Formal deformation}\label{Def}

Let $(\mathcal{A}, d)$ be a DG-algebra with $\mathcal{A}^p=0$ for $p<0$ and let $H^\bullet(\mathcal{A})$ denote the corresponding cohomology groups. Then the intersection  $A=\mathcal{A}^0\cap \ker d\simeq H^0(\mathcal{A})$ is a differential graded subalgebra of $\mathcal{A}$ (concentrated in  zero degree and endowed with trivial  differential). The enveloping algebra $A^e=A\otimes A^{\mathrm{op}}$  of $A$ acts naturally on $\mathcal{A}$:
$$
(a\otimes b) c=acb\,,\qquad \forall a,b\in A\,, \quad\forall c\in   \mathcal{A}\,. 
$$
This allows us to think of $(\mathcal{A}, d)$ as a cochain complex of (left) modules over $A^e$. Let us further assume that  $H^p (\mathcal{A})=0$ for all $p>0$. Then,                                                    
\begin{equation} \label{I-res}
\xymatrix{0\ar[r] &A\ar[r]^{\varepsilon}&\mathcal{A}^0\ar[r]^{d}& \mathcal{A}^1\ar[r]^d& \cdots}
\end{equation}
is a coresolution of the $A^e$-module $A$.  Denote by $Z(\mathcal{A})$ the center of the algebra $\mathcal{A}$ and let $H^\bullet(\mathcal{A}^A)$ denote the cohomology groups of the subcomplex $\mathcal{A}^A\subset \mathcal{A}$ of $A$-invariant cochains. The main result of the paper can now be  formulated as follows.

\begin{theorem}\label{Th}
Any cohomology class  $[\lambda]\in H^2(\mathcal{A}^{A})$ with a representative   $\lambda \in Z(\mathcal{A})$ defines an integrable deformation of $A$.
\end{theorem}

The rest of this section will be devoted to the proof of the above theorem. Our proof is constructive, namely, we will derive explicit recurrent formulas that 
allow one to find the corresponding $\ast$-product up to any given order. 

First of all, we extend the algebra $\mathcal{A}$ to the DG-algebra $\mathcal{A}[[t]]$  of formal power series in deformation parameter $t$; the action of the operator $d$ extends to $\mathcal{A}[[t]]$ by $k[[t]]$-linearity and continuity.  Correspondingly, the cochains of the bicomplex $C^{\bullet,\bullet}(\mathcal{A}[[t]],\mathcal{A}[[t]])$ are  assumed to be $k[[t]]$-multilinear and continuous with respect to the $t$-adic topology. This ensures the isomorphism  
\begin{equation}\label{FD}
C^{\bullet,\bullet}(\mathcal{A}[[t]], \mathcal{A}[[t]])
\simeq C^{\bullet,\bullet}(\mathcal{A}, \mathcal{A})[[t]]\,.
\end{equation}
In other words, each cochain  of $C^{\bullet,\bullet}(\mathcal{A}[[t]],\mathcal{A}[[t]])$ is completely specified by its values on the subalgebra $\mathcal{A}\subset \mathcal{A}[[t]]$. This allows us to view the  cochains as spanning a $k$-vector space 
graded by powers of $t$. Then the homogeneous subspaces  are precisely the eigenspaces  of the Euler operator $N=t\frac{d}{dt}$. Associated to this grading is a descending filtration: we say that the (filtration) degree of $f$ is bigger than $m$, if $f\in t^{m+1} C^{\bullet,\bullet}(\mathcal{A},\mathcal{A})[[t]]$.  In this case we write
$\deg f> m$.  The action of the Euler operator on the cochains $f\in C^{\bullet,\bullet}(\mathcal{A}[[t]],\mathcal{A}[[t]])$ is given by  
$$
(N f)(a_1\otimes \cdots\otimes a_n)= N f(a_1\otimes\cdots\otimes a_n)-\sum_{i=1}^n  f(a_1\otimes\cdots\otimes Na_i\otimes\cdots\otimes a_n)\,.
$$
It is not hard to see that 
$$
[N,d]=0\,,\qquad [N,\delta]=0\,,\qquad N[f,g]=[Nf,g]+[f,Ng]
$$
for all $f,g\in C^{\bullet,\bullet}(\mathcal{A}[[t]],\mathcal{A}[[t]])$.

Now, starting with the total differential $D=d+\delta$ in $C^{\bullet,\bullet}(\mathcal{A}[[t]],\mathcal{A}[[t]])$, we consider its formal deformation $D_\mu$ defined by 
$$
D_\mu a= Da +[\mu, a]\,, \qquad \forall a\in \mathcal{A}[[t]] \,,  
$$
for some cochain  $\mu\in C^{2,0}(\mathcal{A}[[t]],\mathcal{A}[[t]])$ with 
\begin{equation}\label{degmu}
\deg \mu>0\,.
\end{equation}  It is straightforward to see that the operator $D_\mu$ squares to zero iff the cochain $\mu$ satisfies the Maurer--Cartan equation 
\begin{equation}\label{mceq}
D\mu=-\frac12[\mu,\mu]
\end{equation}
or, what is the same,  
$$
d\mu =0\,,\qquad \delta\mu = -\frac12[\mu,\mu]\,. 
$$
The first equation just says that the cochain $\mu$, when restricted to the subalgebra $A[[t]]\subset \mathcal{A}[[t]]$, takes the values in $A[[t]]$.  
Then the second equation allows us to regard  $\mu$ as an associative deformation of the original multiplication in $A[[t]]$. In other words, the operator 
$\delta+[\mu,-]$ is a differential whenever it defines or comes from an associative product $m+\mu$ in $A[[t]]$.

Our task now is to construct a solution to Eq. (\ref{mceq}) by some class $[\lambda]\in H^2(\mathcal{A}^A)$ with a representative $\lambda\in Z(\mathcal{A})$. To this end, we introduce the  cochains $\Lambda$, $\Psi$ and $\Phi$ of the bicomplex $C^{\bullet,\bullet}(\mathcal{A}[[t]],\mathcal{A}[[t]])$ with the following distributions of bidegrees and filtration degrees:
\begin{equation}\label{deg}
\begin{array}{lll}
\mathrm{Deg}\, \Lambda =(0,2)\,,&\qquad \mathrm{Deg} \,\Psi=(0,1)\,,&\qquad  \mathrm{Deg}\, \Phi=(1,0)\,,\\[3mm]
\deg (\Lambda -\lambda)>0\,,&\qquad \deg \Psi>0\,,&\qquad \deg \Phi>0\,.
\end{array}
\end{equation}
It is also convenient to combine $\Psi$ and $\Phi$ into a single cochain $\Gamma=\Psi+\Phi$ of total degree $1$. Besides the  grading conditions, these cochains are supposed to satisfy the following `master equations':
\begin{equation}\label{BEq}
     D_\mu\Gamma=t\Lambda-N\mu\,,\qquad N\Lambda=[\Gamma,\Lambda]\,.
\end{equation}

By assumption, the complex (\ref{I-res}) is acyclic and splits over $k$. Hence, there is a contracting homotopy operator $h$. This operator extends to the complex $0\rightarrow A[[t]]\stackrel{\varepsilon}{\rightarrow}\mathcal{A} [[t]]$ by $k[[t]]$-linearity  and continuity with respect to the $t$-adic topology such that 
$$
h\varepsilon=1_{{A}[[t]]}\,,\qquad \varepsilon h+hd=1_{\mathcal{A}^0[[t]]}\,,\qquad hd+dh=1_{\mathcal{A}^p[[t]]}\,,\quad p>0 \,.
$$
Setting
$$
(h f)(a_1\otimes \cdots\otimes a_n)= h f(a_1\otimes\cdots\otimes a_n)\,,
$$
one easily checks that 
$$
hd+dh=1_{C^{\bullet,\bullet}}-\varepsilon\sigma\,,
$$
where $\sigma=h|_{\mathcal{A}^0}$.

\begin{lemma}
Eqs. (\ref{BEq}) have a unique solution satisfying  the additional condition 
\begin{equation}\label{NC}
h\Gamma=0\,.
\end{equation}
\end{lemma}
\begin{proof}
Expanding Eqs. (\ref{BEq}) in homogeneous components with respect to the bidegree,  we get the four equations 
\begin{equation}\label{ml}
          N\mu=-\delta \Phi-[\mu,\Phi]\,,\qquad  N\Lambda=[\Phi,\Lambda]\,,
\end{equation}
\begin{equation}\label{FPs}
    d\Phi=-\delta \Psi-[\mu,\Psi]\,,\qquad    d\Psi=t\Lambda\,. 
\end{equation}
Using the homotopy operator and the normalization condition (\ref{NC}), we can solve Eqs. (\ref{FPs}) for $\Psi$ and $\Phi$ as follows:
\begin{equation}\label{PF}
\Psi = th\Lambda \,,\qquad \Phi =-th(\delta h\Lambda+[\mu, h\Lambda])\,.
\end{equation}
Substitution of these expressions  into (\ref{ml}) yields then the system of ODEs
\begin{equation}\label{ode1}
\begin{array}{l}
\dot{\mu}=\delta h(\delta h\Lambda+[\mu, h\Lambda])+[\mu,h(\delta h\Lambda+[\mu, h\Lambda])]\,,\\[3mm]\dot{\Lambda}=-[h(\delta h\Lambda+[\mu, h\Lambda]),\Lambda]\,,
\end{array}
\end{equation}
where the overdot stands for the derivative in $t$. 
One  can solve these equations by iterations and obtain a unique formal solution $\mu(t)$, $\Lambda(t)$ subject to the initial condition 
$$
\mu(0)=0\,,\qquad \Lambda(0)=\lambda\,,
$$
cf. Eqs. (\ref{degmu}) and (\ref{deg}). It is also seen that the cochains (\ref{PF}) obey the conditions (\ref{deg}) and (\ref{NC}). 

\end{proof}

\begin{lemma}
The cochain $\mu$ defined by Eqs. (\ref{BEq})  satisfies the Maurer-Cartan equation (\ref{mceq}) whenever 
\begin{equation}\label{dl}
D\lambda=0\,.
\end{equation}
\end{lemma}

\begin{proof}

Let us put
$$
T=D_\mu\Lambda \,,\qquad  R=D\mu+\frac12[\mu,\mu]\,.
$$
We are going to show that the cochain $R$ vanishes for any $\mu$ and $\Lambda$ satisfying 
Eqs. (\ref{BEq}). 
Applying the operator $D_\mu$ to both sides of Eqs. (\ref{BEq}), we find  
$$
[R,\Gamma]=tT -NR \,,\qquad NT=[\Gamma, T]\,, 
$$
provided $\Lambda$, $\Gamma$, and $\mu$ obey (\ref{BEq}). This yields the system of ODEs
$$
\dot{ R}=T-\frac1t[R,\Gamma]\,,\qquad \dot{T}=\frac1t[\Gamma,T]\,.
$$
Since $\deg \Gamma >0$ the right hand sides of these equations are regular in $t$. So, the equations have a unique solution $R=0$
 and $T=0$ subject to the  initial conditions
 $$
 R(0)=0\,,\qquad T(0)=D\lambda=0\,.
 $$
\end{proof}
The condition (\ref{dl}) amounts to the equations 
$$
\delta \lambda=0\,,\qquad d\lambda=0\,.
$$
The first equation defines $\lambda$ as an element of  the center of the algebra $\mathcal{A}$, then the second equation allows us to regard  $\lambda$ as  a $2$-cocycle of the complex $(\mathcal{A}^A,d)$. Combining the two Lemmas above yields the proof of Theorem \ref{Th}. 

For the sake of completeness we also present explicit recurrent relations for determining $\mu$ up to any given order in deformation parameter.  These relations are obtained by expanding all the cochains $\mu$, $\Lambda$, $\Phi$, and $\Psi$ as well as the defining equations (\ref{ml}) and (\ref{FPs}) in powers of $t$.  We get
\begin{equation}\label{rec-rel}
    \begin{array}{l}
         \displaystyle \mu_n=-\frac1n\Big(\delta \Phi_n+\sum_{k=1}^{n-1}[\mu_k,\Phi_{n-k}]\Big) \,,\\[4mm]
        \displaystyle  \Phi_n=-h\Big (\delta \Psi_n+\sum_{k=1}^{n-1}[\mu_k,\Psi_{n-k}]\Big)\,,\\[5mm]
        \displaystyle  \Psi_n=h\Lambda_{n-1}\,,\\[3mm]
       \displaystyle   \Lambda_{n}=\frac1{n}\sum_{k=1}^n[\Phi_k,\Lambda_{n-k}]\,,\qquad \forall n>0\,,\\[4mm]
         \Lambda_0=\lambda\,.
    \end{array}
\end{equation}

From these relations one can readily find that the first-order  deformation is governed by the following noncommutative Poisson bracket:
\begin{equation}\label{NCPS}
\mu_1=\delta h\delta h\lambda \,.
\end{equation}
Although the Hochschild cocycle $\mu_1\in \mathcal{C}^{2}(A,A)$ looks like a $\delta$-coboundary, it is not the case. One should keep in mind that the `potential'  $h\delta h\lambda $ for $\mu_1$ is not an element of $\mathcal{C}^{1}(A,A)$, rather it takes values in $\mathcal{A}^0$. 

The expression for the second-order deformation is a bit more cumbersome: 
$$
\mu_2=\frac12\Big([\delta h\delta h\lambda, h\delta h\lambda ]+\delta h[\delta h\delta h\lambda, h\lambda] -\delta h \delta h [h\delta h\lambda,\lambda] \Big)\,.
$$

\begin{remark}
In case $\mathcal{A}$ is an injective resolution of $A$, applying the standard spectral sequence arguments to the bicomplex $C^{\bullet,\bullet}(\mathcal{A},\mathcal{A})$ shows that Rel. (\ref{NCPS}) is a particular manifestation of the quasi-isomorphism $$\tau^\bullet: C^\bullet(\mathcal{A}^A)\rightarrow \mathcal{C}^\bullet(A,A)\,, \qquad \tau^n=(\delta h)^n\,,$$ inducing the (right) isomorphism (\ref{Ext}) in cohomology. In that case, any nontrivial $2$-cocycle $\lambda \in C^2(\mathcal{A}^A)$ gives rise to a nontrivial infinitesimal deformation of the algebra $A$. The additional requirement  $\lambda\in Z(\mathcal{A})$ ensures that  the deformation is integrable. 
\end{remark}

%%%%%%%%%%%%%%%%%%%%%%%%%%%%%%%%%%%%%%%%%%%%%%%%%%%%%%%%%%%%%%%%%%%%%%%
\section{Examples of quantization}\label{Examples}
%%%%%%%%%%%%%%%%%%%%%%%%%%%%%%%%%%%%%%%%%%%%%%%%%%%%%%%%%%%%%%%%%%%%%%%

In this section, we illustrate the above method of noncommutative deformation quantization by several examples. 
Although the examples are well known, it seems instructive to re-examine them all from a single perspective.
%%%%%%%%%%%%%%%%%%%%%%%%%%%%%%%%%%%%%%%%%%%%%%%%%%%%%%%%%%%%%%%%%%%%%%%
\subsection{Moyal $*$-product}\label{Moyal} Let $V$ be an $n$-dimensional vector space over $k$. 
We are interested in associative deformations of the commutative algebra  $A=S(V^\ast)$, the symmetric algebra of the dual space $V^\ast$.  The basic ingredient of our method is an appropriate multiplicative coresolution  of the algebra $A$ in the category of $A$-bimodules. To construct such a  coresolution we introduce the symmetric algebra $S(V)$ of $V$. Upon choosing linear coordinates $\{x^i\}$ on $V$ and $\{p_i\}$ on $V^\ast$, we can make the identification $S(V^\ast)=k[x^1,\ldots, x^n]$.
Then  the $k$-vector space $B=\mathrm{Hom}(S(V^\ast), S(V^\ast))$ is isomorphic to the space of formal power series in $p$'s with coefficients in polynomial functions in $x$'s. The space $B$ is given the structure of a noncommutative associative algebra with respect to the following
Moyal-type multiplication: 
\begin{equation}\label{b-prod}
a\bullet b=a\exp\left({\frac{\stackrel{\leftarrow}{\partial}}{\partial x^i}\frac{\stackrel{\rightarrow}{\partial}}{\partial p_i}}\right)b\,,\qquad \forall a,b\in B\,.
\end{equation}
One can easily see that the $\bullet$-product is well defined on $B$.\footnote{Actually, this product originates from the composition of homomorphisms of $\mathrm{Hom}(S(V^\ast), S(V^\ast))$. If we identify these homomorphisms with pseudo-differential operators acting on polynomials, then the $\bullet$-product is nothing else but the product of their $px$-symbols \cite[Ch. 5]{BSh}. } Finally, we define the algebra $\mathcal{A}=B\otimes \Lambda (V)$, where the second factor is given by the exterior algebra of the space $V$. 
The standard grading on $\Lambda(V)$ induces a grading on $\mathcal{A}=\bigoplus\mathcal{A}^l$, so that the homogeneous subspaces $\mathcal{A}^l$ are spanned 
by the $l$-forms
\begin{equation}\label{om}
\omega =\omega^{i_1\cdots i_l}(x,p)dp_{i_1}\wedge\cdots \wedge dp_{i_l}\,,\qquad \omega^{i_1\cdots i_l}(x,p)\in B\,.
\end{equation}
To make $\mathcal{A}$ into a DG-algebra, we endow it with the differential $d: \mathcal{A}^l\rightarrow \mathcal{A}^{l+1}$ defined by 
\begin{equation}\label{d}
d\omega =\frac{\partial \omega^{i_1\cdots i_l}}{\partial p_j} dp_j\wedge dp_{i_1}\wedge\cdots \wedge dp_{i_l}\,.
\end{equation}
By the Poincar\'e Lemma this differential is acyclic in positive degrees and $H^0(\mathcal{A})\simeq A=S(V^\ast)$. Thus, we get a coresolution 
$0\rightarrow A\stackrel{\varepsilon}{\rightarrow} \mathcal{A}$, with $\varepsilon: A\rightarrow \mathcal{A}^0$ being the inclusion. This coresolution is certainly not injective: according to (\ref{b-prod}) 
$$
a\bullet b=ab\qquad \forall a\in \mathcal{A},\quad \forall b\in A
$$
and the $A$-bimodule $\mathcal{A}$ is not divisible from the right. Nonetheless we can proceed with computation of the $A$-invariant cohomology $H^\bullet(\mathcal{A}^A)$. Clearly, it is enough to check the invariance  only for the generators $\{x^i\}$. We find 
$$
x^i\bullet \omega -\omega \bullet x^i=0\quad \Leftrightarrow \quad \frac{\partial\omega }{\partial p_i}=0\,.
$$
Hence, $\mathcal{A}^A = S(V^\ast)\otimes \Lambda (V)\simeq H^\bullet (\mathcal{A}^A)$. In other words, all the nontrivial cocycles are given by the forms 
\begin{equation}\label{HKR}
\omega =\omega^{i_1\cdots i_l}(x)dp_{i_1}\wedge\cdots \wedge dp_{i_l}\,.
\end{equation}
On the other hand, by the Hochschild--Kostant--Rosenberg theorem 
$$HH^\bullet(S(V^\ast),S(V^\ast))\simeq S(V^\ast)\otimes \Lambda (V)\,,$$
and we infer that our coresolution computes the entire Hochschild cohomology of $S(V^\ast)$.  

Although each $2$-form $\omega=\pi^{ij}(x)dp_i\wedge dp_j$ (= a bivector field on $V$) gives rise to an infinitesimal deformation of $S(V^\ast)$, only a part of these deformations is integrable. The necessary and sufficient condition for integrability is the vanishing of the Schouten bracket $[\omega,\omega]_S$, the result following from the Kontsevich formality theorem \cite{Konts}. In our approach integrability is 
ensured by a more restrictive condition, namely, $\omega\in Z(\mathcal{A})$. It is easy to see that a form $\omega$ belongs to the center of $\mathcal{A}$ iff 
$$
\omega =\omega^{i_1\cdots i_l}dp_{i_1}\wedge\cdots \wedge dp_{i_l}\,,\qquad  \omega^{i_1\cdots i_l}\in k\,.
$$
Hence, all the constant bivectors $\lambda=\pi^{ij}dp_i\wedge dp_j$ give rise to integrable deformations.  
Using the standard contracting homotopy $h: \mathcal{A}^{l}\rightarrow \mathcal{A}^{l-1}$, 
\begin{equation}\label{h}
h\omega=\int_0^1 {ds}s^{l-1}\omega^{i_1\cdots i_l}(x,sp)p_{i_1}dp_{i_2}\wedge\cdots\wedge dp_{i_l}\,,
\end{equation}
one readily finds that the first-order deformation is given by the Poisson bracket 
$$
\mu_1(a,b)=\frac12\pi^{ij}\frac{\partial a}{\partial x^i}\frac{\partial b}{\partial x^j}\,,\qquad \forall a,b\in k[x^1,\ldots,x^n]\,.
$$
The whole deformation, being constructed by formulas (\ref{rec-rel}), reproduces the Moyal $\ast$-product.

%%%%%%%%%%%%%%%%%%%%%%%%%%%%%%%%%%%%%%%%%%%%%%%%%%%%%%%%%%%%%%%%%%%%%%%
\subsection{Symplectic reflection algebras} 
%%%%%%%%%%%%%%%%%%%%%%%%%%%%%%%%%%%%%%%%%%%%%%%%%%%%%%%%%%%%%%%%%%%%%%%
We keep the notation of the previous subsection. Let $\Gamma$ be a finite group acting faithfully and linearly on $V$. Now we are interested in formal deformations of the smash-product algebra  $A_\Gamma=S(V^\ast)\rtimes \Gamma$.  As a $k$-vector space $A_\Gamma$ is equal to $S(V^\ast)\otimes k[\Gamma]$ and the multiplication is defined by  
$$
\gamma \cdot f=  {}^\gamma f \cdot \gamma \,,\qquad \forall f\in V^\ast,\quad \forall \gamma \in \Gamma\,,
$$
where ${}^\gamma f(v)=f(\gamma^{-1}v)$ for all $v\in V$. The linear representation of $\Gamma$ in $V$ induces a representation in the space $\mathcal{A}=B\otimes \Lambda (V)$ constituted  by the forms (\ref{om}). Furthermore, the $\bullet$-product (\ref{b-prod}) is obviously invariant under the action of $\Gamma$ and we can define the smash-product algebra $\mathcal{A}_{\Gamma}=\mathcal{A}\rtimes \Gamma$. The differential (\ref{d}) extends to $\mathcal{A}_\Gamma$ by setting $d\gamma=0$ for all $\gamma\in \Gamma$. Thus, $\mathcal{A}_\Gamma$ is a DG-algebra extending  the subalgebra  $A_\Gamma=\ker d \cap \mathcal{A}_\Gamma^0$. 

As an $A$-bimodule the algebra $\mathcal{A}_\Gamma$ splits into the direct sum $\mathcal{A}_\Gamma=\bigoplus_{\gamma\in \Gamma} \mathcal{A}_\gamma$ of submodules,  so that the generic element $\omega\in \mathcal{A}_\Gamma$ expands as  
\begin{equation}\label{k-f}
\omega=\sum_{\gamma\in \Gamma}\omega_\gamma \gamma\,,\qquad \omega_\gamma\in \mathcal{A}_\gamma\,.
\end{equation}
Each module $\mathcal{A}_\gamma$ is isomorphic to $\mathcal{A}$ as a vector space, but the right action of $A$ is now twisted by the element $\gamma$, i.e.,
$$
(a\otimes b)c=a\bullet c\bullet {}^\gamma b\,,\qquad \forall a,b\in A\,,\quad c\in \mathcal{A}_\gamma\,.
$$
Clearly, the $m$-form (\ref{k-f}) is $A$-invariant iff 
$$
x^i\bullet \omega_\gamma -\omega_\gamma\bullet {}^\gamma x^i=0\,,\qquad \forall \gamma\in \Gamma\,, \quad \forall i=1,\ldots, n\,.
$$
This is equivalent to the set of differential equations 
\begin{equation}\label{d-eq}
(x^i-{}^\gamma x^i)\omega_\gamma+\frac{\partial \omega_\gamma}{\partial p_i}=0\,. 
\end{equation}
Let us introduce the exponential functions 
\begin{equation}\label{E}
 E_\gamma=e^{-\langle p, x-{}^\gamma x\rangle }\,,
\end{equation}
where the triangle brackets stand for the canonical pairing between the spaces $V$ and $V^\ast$. 
Then the general solution to Eqs. (\ref{d-eq}) can be written as 
\begin{equation}\label{om-g}
\begin{array}{c}
\omega_{\gamma}=E_\gamma \bar\omega_\gamma \,,\\[3mm]\forall  \,\bar\omega_\gamma =\bar\omega_\gamma^{i_1\cdots i_m}(x)dp_{i_1}\wedge\cdots\wedge dp_{i_m} \in S(V^\ast)\otimes \Lambda^m(V)\,.
\end{array}
\end{equation}
Verifying the $\Gamma$-invariance of the form (\ref{k-f}),  we find
$$
 {}^\alpha E_\gamma=\alpha E_\gamma \alpha^{-1}=E_{\alpha\gamma \alpha^{-1}}
$$
and 
$$
{}^\alpha\omega=\alpha\omega \alpha^{-1}=\sum_{\gamma\in \Gamma}E_{\alpha\gamma \alpha^{-1}} {}^\alpha\bar\omega_\gamma \alpha\gamma \alpha^{-1}=\sum_{\gamma\in \Gamma}E_{\gamma}  {}^\alpha\bar\omega_{\alpha^{-1}\gamma \alpha} \gamma\,.
$$
Since the exponential functions $E_\gamma$ are linearly independent over $S(V^\ast)\otimes \Lambda(V)$, we conclude that ${}^\alpha\omega=\omega$ iff 
$$
{}^\alpha\bar\omega_\gamma=\bar\omega_{\alpha\gamma \alpha^{-1}}\qquad \forall \gamma, \alpha\in \Gamma\,.
$$
If we treat the assignment $\gamma\mapsto \bar\omega_\gamma$ as a map from $\Gamma$ to $S(V^\ast)\otimes \Lambda(V)$, then the last relation identifies this map as $\Gamma$-equivariant with respect to the adjoint action of the group  $\Gamma$ on itself.

\begin{remark}
Notice that the exponentials (\ref{E}) define an antihomomorphism from $\Gamma$ to the group of invertible elements of the $\bullet$-product algebra $B$.  Indeed,  
$$
\begin{array}{c}
E_\gamma\bullet E_{\beta}=e^{-\langle p,x-{}^\gamma x \rangle } e^{\langle\frac{\stackrel{\leftarrow}{\partial}}{\partial x},\frac{\stackrel{\rightarrow}{\partial}}{\partial p}\rangle}e^{-\langle p, x-{}^{\beta} x \rangle }\\[3mm]
=e^{-\langle p,(x+ \frac{\stackrel{\rightarrow}{\partial}}{{\partial} p})-{}^\gamma (x+ \frac{\stackrel{\rightarrow}{\partial}}{{\partial} p})\rangle } e^{-\langle p,x-{}^{\beta}x \rangle }\\[3mm]
=e^{-\langle p,x-{}^{\beta\gamma} x\rangle } =E_{\beta\gamma}\,.
\end{array}
$$
Combining this antihomomorphism with inversion on $\Gamma$, one gets the homomorphism $\gamma\mapsto E_{\gamma^{-1}}$. Furthermore, all the $\Gamma$-automorphisms of the $\bullet$-product algebra $B$ turn out to be  internal: 
$$
{}^\gamma a = E_{\gamma^{-1}}\bullet a \bullet E_\gamma\,,\qquad \forall a\in B\,.
$$
The last property is enough to check only for the generators $x^i$ and $p_i$. 

\end{remark}

Our  computations above show that the DG-algebra $\mathcal{A}_\Gamma^{A_\Gamma}$, when considered as a cochain complex, splits into the direct sum of subcomlexes:
$$
\mathcal{A}_\Gamma^{A_\Gamma} = \bigoplus_{{[\gamma]}\in C(\Gamma)} \big(S(V^\ast)\otimes \Lambda(V)\big)^{Z(\gamma)}\,.
$$
Here $C(\Gamma)$ denotes the space of conjugacy classes of $\Gamma$ and we fix a representative ${\gamma}\in [\gamma]$ in each class $[\gamma] \in C(\Gamma) $; $Z(\gamma)$ stands for the centralizer of the representative $\gamma\in\Gamma$. The elements of the subcomplex $\Omega_\gamma=\big(S(V^\ast)\otimes \Lambda(V)\big)^{Z(\gamma)}$ are given by $Z(\gamma)$-invariant forms (\ref{om-g}). The differential in $\mathcal{A}_\Gamma$ induces the following coboundary operator in 
each $\Omega_\gamma$:
$$
\partial_\gamma \bar\omega_\gamma=-\langle dp,x-{}^\gamma x\rangle\wedge \bar\omega_\gamma\,.
$$
The cohomology of $\partial_\gamma$ is easily computed by decomposing the carrier space into the direct sum $V=V^\gamma\oplus N^\gamma$, where 
$V^\gamma=\ker (1-\gamma)|_V$ and $N^\gamma=\mathrm{im} (1-\gamma)|_V$. We put $l(\gamma)=\dim N^\gamma$. Clearly, $l(\alpha\gamma \alpha^{-1})=l(\gamma)$ for all $\gamma,\alpha\in \Gamma$. Then it is not hard to find that\footnote{See e.g. \cite[Sec. 3]{NPPT} or \cite[Prop. 4.2]{ShSk} for similar computations.}  
$$
H^m(\Omega_\gamma, \partial_\gamma)=\Big( S({V^\gamma}^\ast)\otimes \Lambda^{m-l(\gamma)}({V^\gamma})\otimes \Lambda^{l(\gamma)}({N^\gamma})    \Big)^{Z(\gamma)} \,.
$$
This agrees with the results of Ref. \cite{NPPT}, where it was shown that 
$$
HH^\bullet(S(V^\ast)\rtimes \Gamma, S(V^\ast)\rtimes \Gamma)\simeq\bigoplus_{[\gamma]\in C(\Gamma)}\big(S({V^\gamma}^\ast)\otimes \Lambda^{\bullet-l(\gamma)}({V^\gamma})\otimes \Lambda^{l(\gamma)}({N^\gamma})\Big)^{Z(\gamma)}.
$$

Thus, we see that our coresolution $\mathcal{A}_\Gamma$ computes the entire  Hochschild cohomology of the smash-product algebra $A_\Gamma$. 
For the second cohomology group the structure of the direct product above can be refined. Notice that $\Lambda^1({N^\gamma})={N^\gamma}$
and for $l(\gamma)=1$  the $1$-dimensional space $\Lambda^{l(\gamma)}({N^\gamma})$ has no $\gamma$-invariant vectors by the definition of $N^\gamma$. Hence, the group elements $\gamma\in \Gamma$ with $l(\gamma)=1$ do not contribute to the Hochschild cohomology.  Furthermore, when the action of $\Gamma$ on $V$ is faithful\footnote{The case of a nonfaithful action was considered in \cite{SW-1}.}, the identity $e\in \Gamma$  is the only group element with $l(\gamma)=0$. Therefore, we conclude
$$
HH^2(A_\Gamma,A_\Gamma)=\big( S(V^\ast)\otimes \Lambda^2(V)\big)^{\Gamma}\oplus \Big(\bigoplus_{\gamma\in \Gamma, l(\gamma)=2} S({V^\gamma}^\ast)\otimes \Lambda^2({N^\gamma})\Big)^{\Gamma}.
$$
This is the space  of infinitesimal deformations of the algebra $A_\Gamma$. Having in mind integrable deformations, we are looking for nontrivial $2$-cocycles belonging to the center of the algebra $\mathcal{A}_\Gamma$. These are given by the $\Gamma$-invariant forms (\ref{k-f}) and (\ref{om-g}) satisfying the additional conditions:
$$
p_i\bullet \omega -\omega\bullet p_i=0\quad \Leftrightarrow\quad \frac{\partial \bar\omega_\gamma}{\partial x^i}=0\,,\quad \forall \gamma\in \Gamma\,.
$$
In other words, the coefficients of the central forms (\ref{k-f}) depend on $x$'s and $p$'s only through the exponential multipliers (\ref{E}).  

Consider now a special case that $V$ is a symplectic vector space and $\Gamma\subset \mathrm{Sp}(V)$. Let $\pi$ denote the Poisson bivector on $V$ dual to the symplectic structure. An element $\gamma\in \mathrm{Sp}(V)$  is called a symplectic reflection if $l(\gamma)=2$. Correspondingly,  $\Gamma$ is a symplectic reflection group if it is generated by symplectic reflections \cite{EG}. Denote by $S_\Gamma$ the set of all symplectic reflections of $\Gamma$. Starting from the nondegenerate Poisson  structure $\pi$, we can define the following $2$-cocycle of the algebra $\mathcal{A}_\Gamma$:
$$
\lambda =\omega_e+\sum_{\gamma\in S_\Gamma}c(\gamma)\omega_\gamma \gamma \,,
$$
$$
\omega_e=\pi(dp,dp)=\pi^{ij}dp_i\wedge dp_j \in \Lambda^2(V)^\Gamma\,,
$$
$$
\omega_\gamma=E_\gamma \bar \omega_\gamma\,,\qquad \bar\omega_\gamma=\pi(dp-d^\gamma p, dp-d^\gamma p) =\pi_{\gamma}(dp,dp)\in \Lambda^2(N^\gamma)^{Z(\gamma)}\,.
$$
The $\Gamma$-invariance of $\lambda$ requires that $c(\alpha \gamma \alpha^{-1})=c(\gamma)$ for all $\alpha,\gamma\in \Gamma$, that is, $c(\gamma)$ is a class function. The corresponding Poisson structure on $\mathcal{A}_\Gamma$ can now be restored by the general formula (\ref{NCPS}). Using the contracting homotopy (\ref{h}), we find 
$$
(h\delta h\omega_\gamma)(a) =h\left[ \Big(\frac{\partial a}{\partial x^i}\Big)\Big(x+\frac{\partial}{\partial p}\Big)\int_0^1ds s e^{-s\langle p, x-{}^\gamma x\rangle }
\pi^{ij}_\gamma dp_j\right] 
$$
$$
=h\int_0^1 ds s\Big(\frac{\partial a}{\partial x^i}\Big)\big((1-s)x+ s^\gamma x \big)e^{-s\langle p, x-{}^\gamma x\rangle }
\pi^{ij}_\gamma dp_j
$$
$$
=\int_0^1 du \int_0^1 ds s\Big(\frac{\partial a}{\partial x^i}\Big)\big((1-s)x+ s^\gamma x \big)e^{-su\langle p, x-{}^\gamma x\rangle }
\pi^{ij}_\gamma p_j
$$
and then
$$
(\delta h\delta h \omega_\gamma)(a,b)
$$
$$
\left .= b\Big(x+\frac{\partial}{\partial p}\Big) \int_0^1 du \int_0^1 ds s\Big(\frac{\partial a}{\partial x^i}\Big)\big((1-s)x+ s^\gamma x \big)e^{-su\langle p, x-{}^\gamma x\rangle }
\pi^{ij}_\gamma p_j\right|_{p=0}
$$
$$
=\int_0^1 du \int_0^1 ds s\Big(\frac{\partial a}{\partial x^i}\Big)\big((1-s)x+ s^\gamma x \big)\Big(\frac{\partial b}{\partial x^j}\Big)\big((1-su)x+ su^\gamma x \big)
\pi^{ij}\,.
$$
Here we simplified  our calculations by noticing that, by construction, $\mu_1(a,b)$ must be independent of $p$'s, so that  we can put $p_i=0$. 
After the change  of variables $w=su$, we can write the Poisson structure $\mu_1$ through the integrals over a $2$-simplex: 
\begin{equation}\label{SRA}
\mu_1(a,b)=(\delta h\delta h\lambda)(a,b)=\frac12\frac{\partial a}{\partial x^i}\frac{\partial b}{\partial x^j}\pi^{ij}
\end{equation}
$$
+\sum_{\gamma\in S_\Gamma }c(\gamma)\int\limits_{0<w<s<1} dwds\Big(\frac{\partial a}{\partial x^i}\Big)\big((1-s)x+ s^\gamma x \big)\Big(\frac{\partial b}{\partial x^j}\Big)\big((1-w)x+ w^\gamma x \big)
\pi^{ij}_\gamma \gamma\,.
$$

To the best of our knowledge this representation for the noncommutative Poisson structure on $A_\Gamma=S(V^\ast)\rtimes \Gamma$ is new, cf. \cite[Sec. 3.1]{HT}.  Thus, any constant (i.e., of polynomial degree zero) 2-cocycle of the algebra ${A}_\Gamma$ defines a noncommutative Poisson structure.
This is in agreement with the general results of Ref. \cite[Cor. 8.2]{SW-3}.  
The higher-order deformations can be found in a similar way by formulas (\ref{rec-rel}). The resulting algebra is the formal analog of  the symplectic reflection algebra introduced by Etingof and Ginzburg \cite{EG}. It is a $(p+1)$-parameter deformation of the algebra $A_\Gamma$, with $p$ being the number of conjugacy classes in $S_\Gamma$. 

\subsection{Smash products of the Weyl algebra} In case $c(\gamma)=0$, the sum (\ref{SRA}) reduces to the canonical Poisson bracket on $V$, whose deformation quantization gives the Moyal $\ast$-product of Sec. \ref{Moyal}. Actually, the Moyal deformation of the polynomial algebra $S(V^\ast)$ is not formal and one may put the deformation parameter $t$ to be an arbitrary number, say $t=2$. This gives the space $S(V^\ast)$ the structure of a noncommutative, associative algebra $\mathrm{Weyl}(V)$, called the polynomial Weyl algebra of the symplectic vector space $V$. The group $\Gamma\subset \mathrm{Sp}(V)$ acts naturally on $\mathrm{Weyl}(V)$ by automorphisms and we can form the smash-product algebra $\mathrm{Weyl}(V)\rtimes \Gamma$. It follows from the above consideration  that the algebra $\mathrm{Weyl}(V)\rtimes \Gamma$ admits the $p$-parameter deformation associated to the functions $c(\gamma)$ in Eq. (\ref{SRA}). Moreover, by the Alev--Farinati--Lambre--Solotar (AFLS) theorem \cite[Sec. 6.1]{AFLS} this deformation exhausts all the possibilities. One could try to obtain this deformation directly starting from the Weyl algebra. The main problem would then to find an appropriate coresolution. Actually, such a coresolution has been already constructed in our recent paper \cite{ShSk}. As a $k$-vector space it is given by the space $\mathcal{A}_\Gamma$ of the previous subsection, while the  $\bullet$-product (\ref{b-prod}) is now replaced by 
$$
a\bullet b = a e^{\frac{\stackrel{\leftarrow}{\partial}}{\partial x^i}\Big (\frac{\stackrel{\rightarrow}{\partial}}{\partial p_i} + \pi^{ij}\frac{\stackrel{\rightarrow}{\partial}}{\partial x^j}\Big)} b\,,\qquad \forall a,b\in B\,.
$$
In accordance with the AFLS theorem the central $2$-cocycles of the DG-algebra $\mathcal{A}_\Gamma$ are given by the linear combinations \cite{ShSk}
$$
\lambda = \sum_{\gamma\in S_\Gamma}c(\gamma)E_\gamma \pi_\gamma(dp,dp)\gamma \,,
$$
where $c(\gamma)$ is a class function on $\Gamma$ and 
$$
 E_\gamma=e^{-\langle p, x-{}^\gamma x\rangle +\pi(p,{}^\gamma p)}\,.
$$
Again, one can easily check that $E_\gamma\bullet E_\beta=E_{\beta\gamma}$. Omitting intermediate computations, which are similar to those of the previous subsection, we just write down the final expression for the noncommutative Poisson bracket: 
\begin{equation}\label{mmm}
\mu_1(a,b)=(\delta h\delta h \lambda)(a,b)
\end{equation}

$$
=\sum_{\gamma\in S_\Gamma}c(\gamma)\int\limits_{0<w<s<1} dwds   e^{ \langle p_1, (1-s)x+s^\gamma x \rangle +\langle p_2, (1-w)x+w^\gamma x \rangle} $$
$$
\times e^{\pi(p_2,p_1) +s\pi(p_1,p_2-{}^\gamma p_2-{}^\gamma p_1) +w\pi(p_2,p_1-{}^\gamma p_1-{}^\gamma p_2)}
$$
$$
\left.\times e^{s^2\pi(p_1,{}^\gamma p_1)+sw[\pi(p_1,{}^\gamma p_2)+\pi(p_2,{}^\gamma p_1)]+w^2\pi(p_2,{}^\gamma p_2)}
\pi_\gamma (p_1,p_2)a(x_1)b(x_2) \gamma \right|_{x_1=x_2=0}\,.
$$
Here  
$$
p_1=\left\{\frac{\partial}{\partial x_1^i}\right\}\,, \qquad p_2=\left\{\frac{\partial}{\partial x_2^i}\right\}\,.
$$
The exponential function in the above integral is to be expanded in the Taylor series and integrated term by term. As the functions $a(x_1)$  and $b(x_2)$ are polynomial, only finitely many terms contribute nontrivially to $\mu_1$. 

In the special case that $V$ is a two-dimensional symplectic space and $\Gamma$ is generated by the parity automorphism, ${}^\gamma x=-x$, Rel. (\ref{mmm}) reproduces the Feigin--Felder--Shoikhet $2$-cocycle of Ref. \cite{FFS}. 

\subsection*{Acknowledgments}
We are grateful to Xiang Tang for a useful correspondence. We also thank the anonymous referee for many valuable remarks.

% \subsection*{Acknowledgments} We are grateful to each other for a fruitful collaboration :)

\end{document}